\documentclass[12pt,letterpaper]{article}
\usepackage{amssymb,amsmath,amsthm,enumerate,nicefrac}
\usepackage{graphicx, color}
    \usepackage[driverfallback=hypertex,pagebackref=true,colorlinks]{hyperref}
    \hypersetup{linkcolor=[rgb]{.7,0,0}}
    \hypersetup{citecolor=[rgb]{0,.7,0}}
    \hypersetup{urlcolor=[rgb]{.7,0,.7}}
\usepackage{geometry}
\usepackage{epstopdf}
\usepackage{float}
\floatstyle{ruled}
\newfloat{algorithm}{tbp}{loa}
\providecommand{\algorithmname}{Algorithm}
\floatname{algorithm}{\protect\algorithmname}

\usepackage[titletoc,title]{appendix}
\usepackage{url}

\newtheorem{theorem}{Theorem}
\newtheorem{lemma}{Lemma}[section]

\newtheorem{corollary}[theorem]{Corollary}

\def\N{{\mathbb {N}}}

\def\F{{\mathbb {F}}}

\def\R{{R}}

\DeclareMathOperator{\rank}{rank}

\title
{Polynomial Lower Bounds for 
Arithmetic Circuits over 
Non-Commutative Rings}
\author{Ran Raz
\thanks{Department of Computer Science, Princeton University. Research supported by a Simons Investigator Award.
Email: \texttt{ran.raz.mail@gmail.com}}}
\date{}
\begin{document}
\maketitle

\begin{abstract}
We prove a lower bound of $\Omega\left(n^{1.5}\right)$ for the number of product gates in
non-commutative arithmetic circuits for an explicit $n$-variate degree-$n$ polynomial $f_{n}$ 
(over every field).

We observe that this implies
that over certain non-commutative rings $\R$, any arithmetic circuit that computes the induced polynomial function $f_{n}: \R^n \rightarrow \R$,
using the ring operations of addition and multiplication in $\R$, requires 
at least $\Omega\left(n^{1.5}\right)$ multiplications.

More generally, for any $d\geq 2$ and sufficiently large $n \in \N$,  we obtain a lower bound of $\Omega\left(d\sqrt{n}\right)$ for $n$-variate degree-$d$ polynomials, for both these models.

Prior to our work, the only known lower bounds for the size of non-commutative circuits,  or for the size of arithmetic circuits over any ring, were slightly super-linear in $\max\{n,d\}$: $\Omega\left(n\log d\right)$ by Baur and Strassen~\cite{Str73,BS83}, and $\Omega\left(d\log n\right)$ by Nisan~\cite{Nis91}\footnote{This bound was proved for non-commutative arithmetic circuits and implies a bound for arithmetic circuits over non-commutative rings by our observation.}.

\end{abstract}

\section{Introduction}

Arithmetic circuits are the standard computational model for computing polynomials, such as the determinant or the permanent of a matrix. While arithmetic circuits are usually defined over fields, they are often defined and studied more generally, over rings.
Given a ring $\R$ and an $n$-variate polynomial function $f(x_1,\ldots,x_n)$ over~$\R$, 
a fundamental question is: what is the minimal number of  $+, \times$ ring operations in $\R$,  needed to compute $f$?

The only known lower bounds for the size of general arithmetic circuits over fields, for explicit  $n$-variate polynomials of degree $d$, are of the form $\Omega\left(n\log d\right)$, first established by Strassen~\cite{Str73} (for circuits with $n$ outputs) and by Baur and Strassen~\cite{BS83} (for circuits with one output). 
Improving these bounds is a major, long-standing open problem.
Over the half-century since these landmark results, researchers have also studied a variety of restricted models of arithmetic circuits.

One of the earliest and most extensively studied restricted models is that of non-commutative circuits.
Non-commutative circuits are the standard computational model for computing non-commutative polynomials, that is, polynomials in which the variables do not commute.
Formally, a non-commutative polynomial in variables $x_1,\ldots,x_n$, over a field $\mathbb{F}$, is a formal linear combination of words over the alphabet $\{x_1,\ldots,x_n\}$, with coefficients in~$\mathbb{F}$.
The sum of two non-commutative polynomials is their sum as linear combinations over~$\mathbb{F}$, and their product is defined by defining the product of two non-commutative monomials (words over the alphabet $\{x_1,\ldots,x_n\}$) as their concatenation and extending bilinearly to  all pairs of non-commutative polynomials.
A non-commutative arithmetic circuit is defined similarly
to a standard arithmetic circuit, except that the inputs to each product gate are ordered, and the gate multiplies them in that order.

Interest in non-commutative computation goes back to the early seventies~\cite{Win70,HK71}. Lower bounds for the size of non-commutative circuits were first studied by Hyafil~\cite{Hya77}.
In~1991, Nisan proved a remarkable exponential lower bound of $n^{\Omega(d)}$ for the size of non-commutative formulas, which implies a lower bound of  $\Omega\left(d\log n\right)$ for the size of non-commutative circuits~\cite{Nis91}. 

\subsection{Motivation}

Beyond the historical interest, the literature offers several motivations for studying non-commutative arithmetic circuits. 
First, many arithmetic computations of interest involve objects that do not commute, such as matrices, making non-commutative circuits a natural model to study.
Second, as with other restricted models, one could hope  that understanding the power of non-commutative circuits will shed light on the power of general arithmetic circuits. In particular, proving lower bounds in the non-commutative setting is easier than proving lower bounds for general arithmetic circuits over fields, and thus provides a natural and challenging intermediate goal.
Finally, comparing the relative power of non-commutative and general arithmetic circuits in order to understand the computational advantage of commutativity in arithmetic computations is an interesting goal in its own right.

From our perspective, an additional central motivation for investigating lower bounds for non-commutative circuits is their direct applicability to arithmetic circuits over non-commutative rings. 
When considering arithmetic circuits over non-commutative rings, the non-commutativity of the model is not an imposed restriction, but rather an inherent property of the underlying algebraic structure, as the ring elements do not always commute. 

Specifically, we observe that any lower bound for non-commutative arithmetic circuits over fields implies a corresponding lower bound for arithmetic circuits over certain non-commutative rings. In particular, our lower bound for  non-commutative circuits implies a lower bound of $\Omega\left(n^{1.5}\right)$ for the number of multiplications, and hence also for the total number of ring operations, needed to compute 
an explicit $n$-variate degree-$n$ polynomial function from $\R^n$ to $\R$, 
where $\R$ is a certain ring. 

We find this connection significant since, to the best of our knowledge, these are the first lower bounds for the size of arithmetic circuits over any ring, beyond the $\Omega\left(n\log d\right)$ lower bounds of Baur and Strassen~\cite{Str73,BS83}, and the $\Omega\left(d\log n\right)$ lower bound that can be obtained by Nisan's result~\cite{Nis91} (via the connection that we observe here).

Although this connection between lower bounds for non-commutative arithmetic circuits and lower bounds for arithmetic circuits over non-commutative rings is straightforward, we are not aware of any prior work in which it is stated,  
proved, or proposed as a motivation for studying non-commutative circuits. 
We note however that a related viewpoint is implicit in previous works, such as~\cite{CS07,AS18,CHSS11}.

Specifically, 
Chien and Sinclair consider the function computed by an arithmetic branching program over a field $\F$ when the inputs are taken from  an $\F$-algebra~$A$, rather than from the field $\F$ itself.
They note that when  $A = \F\langle X \rangle $ is the free algebra over~$\F$,
Nisan's lower bound for non-commutative arithmetic branching programs over $\F$ implies a lower bound for computing polynomial functions from  ${A}^n$ to ${A}$ by arithmetic branching programs over~$\F$.
Their main results show similar lower bounds for other $\F$-algebras~\cite{CS07}.

The ring that we consider here is also $\F\langle X \rangle$. The key difference between our observation and that of Chien and Sinclair is that we consider arithmetic circuits over the ring itself, so multiplication by arbitrary ring elements is allowed, whereas Chien and Sinclair consider arithmetic branching programs over the base field 
$\F$, where only elements of 
$\F$ may appear as scalars.

\subsection{Our Results}

For every $d\geq 2$ and sufficiently large $n \in \N$,  we give an explicit $n$-variate non-commutative polynomial $f_{n,d}$ of degree $d$, with coefficients in $\{0,1\}$, such that over every field, any non-commutative circuit computing~$f_{n,d}$ requires at least $\Omega\left(d\sqrt{n}\right)$ non-scalar product gates.

We observe 
that any lower bound for non-commutative arithmetic circuits over fields implies a corresponding lower bound for arithmetic circuits over non-commutative rings.
Specifically,
given a field  $\F$ and a set of variables $Z=\{z_1,\ldots,z_n\}$, we consider the ring 
of non-commutative polynomials in variables $z_1,\ldots,z_n$, over the field $\F$, that is,
the ring $R = \F\langle Z \rangle$. 
We prove that for every 
non-commutative polynomial $f(x_1,\ldots,x_n)$,
if there exists an arithmetic circuit 
over the ring $R$, that computes 
$f(x_1,\ldots,x_n)$ as a function from $R^n$ to $R$,
then there exists a non-commutative arithmetic circuit for $f$ over the field~$\F$, with the same number of sum gates, the same number of product gates, and the same size and depth.

As a consequence, over the ring $R = \F\langle Z \rangle$, any computation of the induced polynomial function $f_{n,d}: \R^n \rightarrow \R$, 
using the ring operations of addition and multiplication in $\R$, requires at least $\Omega\left(d\sqrt{n}\right)$ multiplications. 

We  note that the same proof yields the same lower bound over the ring
$R = \F\langle z_1, z_2, \ldots \rangle$, the free algebra over countably many variables. This may be interesting as this ring is fixed and does not depend on $n$.

In the special case $d=n$, our results give a  lower bound of $\Omega\left(n^{1.5}\right)$, for both these models.
In the special case $d=n^c$, for an arbitrarily large constant $c$, our results give a lower bound of $\Omega\left(n^{c+0.5}\right)$ for both these models.

\subsection{Related Work}

As mentioned above, Nisan established a remarkable exponential lower bound of $n^{\Omega(d)}$ for the size of non-commutative formulas~\cite{Nis91}. This  implies a lower bound of $\Omega\left(d\log n\right)$ for the size of non-commutative circuits. Nisan also explicitly listed 
proving a lower bound for non-commutative circuit size
as an open problem~\cite{Nis91}. While non-commutative circuits have been extensively studied in subsequent works, no better lower bound has been obtained.

Hrubeš, Wigderson and Yehudayoff
initiated a direction for proving exponential lower bounds for the size of non-commutative circuits by connecting the problem to the so-called Sum-of-Squares problem~\cite{HWY11}.
Several works established exponential lower bounds for non-commutative circuits under additional restrictions~\cite{LMS16,LLS19,LMP19,LTS22}.

Two recent works established polynomial
lower bounds
for non-commutative circuits under certain additional restrictions:
Chatterjee and Hrubeš proved a lower bound of 
$\Omega\left(nd\right)$, when $d \leq n$, and 
$\Omega\left(\tfrac{nd\log n}{\log d}\right)$, when $d > n$, for the size of
non-commutative  circuits under the additional restriction that the circuit is homogeneous, that is, under the restriction that the polynomial computed by any subcircuit is a homogeneous polynomial~\cite{CH23}. Shastri proved a lower bound of 
$\Omega\left(n^{1+\epsilon}\right)$, when $d = n$, for the size of
non-commutative  circuits under the (weaker) restriction that the syntactic degree of the circuit is $O(n)$~\cite{Sha26}.

Carmosino, Impagliazzo, Lovett and Mihajlin studied hardness amplification for non-commutative circuits, by reducing the number of input variables, and  proved that lower bounds of $\Omega(n^{\omega/2+\epsilon})$ for the size of non-commutative circuits
for polynomials of constant degree imply exponential lower bounds for non-commutative circuits (where $\omega$ is the best exponent for matrix multiplication)~\cite{CILM18}. We were unable to use their results or techniques to strengthen our results.

Non-commutative circuits have been studied in many additional works. Some of the topics that were investigated are: 
The ${\mathsf{VP}}$ vs. ${\mathsf{VNP}}$ problem in the non-commutative setting~\cite{HWY10};
Non-commutative arithmetic circuits with division gates~\cite{HW15}; The relative hardness of permanent and determinant in the non-commutative setting~\cite{AS18,CHSS11,Bla15,Gen14}; The polynomial identity testing problem in the non-commutative setting~\cite{RS05,BW05,FS13,LMP19,GGOW16,BFGOWW19}.

For excellent introductions to arithmetic circuit complexity and lower bounds for arithmetic circuits, see~\cite{SY10,BCS13,Sap21}.

\subsection{Proof Outline}

\subsubsection{Lower Bounds for Non-Commutative Circuits}

Let $C$ be a non-commutative circuit with  input variables  $x_1,\ldots,x_n$, over a field $\F$. Let~$f$ be the non-commutative polynomial computed by $C$. 
Assume that $f$ is a homogeneous polynomial of degree $d$. Assume that $d$ is even.

We start by modifying the circuit $C$ to a new circuit that computes $f$ and has at most the same number of non-scalar product gates as~$C$. The new circuit will have several additional properties.
First, every node $v$ in the modified circuit computes a polynomial in which all monomials are of degree at least~1, that is, a polynomial without a constant term.
Second, there are no scalar-product gates in the circuit (or scalars at all). Instead, 
every edge to a sum gate is labeled with a field element.
By a standard convention, 
this element multiplies the output of the edge. 
We adopt this  convention to simplify the presentation.
Finally, the $+,\times$ gates in the modified circuit are alternating, along every path from the output gate to an input gate.

The main step of the modification splits every node $v$ in the original circuit into two nodes. One of these nodes computes the degree-zero part of the polynomial computed by $v$ and the other computes the positive-degree part of the polynomial computed by~$v$.
Later on, all constants can be removed.

For a non-commutative polynomial $f$ and  $a,b \in \N$, we define
$M_f^{a,b}$ to be the $n^a \times n^b$ matrix whose entry at
row $(i_1,\ldots,i_a)$ and column $(j_1,\ldots,j_b)$ is the coefficient in $\F$ of the monomial (word)
$x_{i_1}\ldots x_{i_a}x_{j_1}\ldots x_{j_b}$ in the polynomial $f$, where
$i_1,\ldots,i_a, j_1,\ldots,j_b \in \{1,\ldots,n\}$.
This definition goes back to Nisan's breakthrough work~\cite{Nis91}.

Nisan proved his lower bound for non-commutative formulas through a rank argument on the matrix $M_f^{a,b}$.
Similar definitions and rank-based methods have been used by numerous subsequent works, establishing this approach as one of the main techniques for proving lower bounds for arithmetic circuits - see for example~\cite{NW97,Raz09,Raz06,LTS25}.

For a node $v$ in the circuit, 
let $f_v$ be the non-commutative polynomial computed by the node~$v$. 
We denote the matrix $M_{f_v}^{a,b}$  by $M_{v}^{a,b}$.
We analyze the evolution of the rank of $M_{v}^{a,b}$ for different values of $a,b$ along a path from the output gate to an input gate.

We  define a path $(v_0,v_1,\ldots,v_t)$ that goes backward on the circuit, starting from the output gate. 
For every~$i$, we  define $v_{i+1}$ to be one of the children of~$v_i$, according to some specific rules.
For every $v_i$, we  also define $a_i,b_i \in \N$, where
$a_0,b_0 = d/2$. When $v_i$ is a sum gate we will have $a_{i+1} = a_i$ and $b_{i+1} = b_i$ and when  $v_i$ is a product gate we will have either $a_{i+1} = a_i$ and $b_{i+1} = b_i -j$, or $a_{i+1} = a_i -j$ and $b_{i+1} = b_i$, for some $j\geq 1$.

For every $i$, we define,
$$
r_i \; = \; \rank \big( M_{v_i}^{a_i,b_i} \big)
$$
By assumption, we will have $r_0= n^{d/2}$. We will stop when $r_t \leq 1$. Hence, $r_i$ decreases rapidly along the path. On average, each alternation of a sum gate and a product gate needs to reduce~$r_i$ by a factor of roughly $\sqrt{n}$, as there will be at most $d$ such alternations along the path (because when $v_i$ is a product gate, $a_{i+1}+b_{i+1} < a_i+b_i$).

When defining the path $(v_0,v_1,\ldots,v_t)$, we try to minimize the decrease of the value of~$r_i$ along the path, though several other parameters need to be taken into account.
How can the value of $r_i$ decrease rapidly along the path?

If $v_i$ is a sum gate, with children $u_1,\ldots,u_k$,  by subadditivity of the rank, we have,
\begin{equation*}
\rank \left( M_{v_i}^{a_i,b_i} \right) \;\;\leq\;\; \sum_{j=1}^k \rank \left( M_{u_j}^{a_i,b_i} \right)
\end{equation*}
Assume for simplicity of this outline that all terms in the sum on the right-hand side of the equation are equal.
Recall that we will have $v_{i+1} \in \{u_1,\ldots,u_k\}$ and $a_{i+1} = a_i, b_{i+1} = b_i$. Thus, the ratio between $r_i$ and $r_{i+1}$ is at most $k$, and this costs the circuit $k$ gates: $u_1,\ldots,u_k$. Since the ratio needs to be close to $\sqrt{n}$ on average, this typically costs the circuit around~$\sqrt{n}$ gates. This needs to occur close to $d$ times along the path, resulting in a lower bound of $\Omega(d \sqrt{n})$ for the number of gates.
We will choose $v_{i+1} \in \{u_1,\ldots,u_k\}$ such that none of the nodes in $\{u_1,\ldots,u_k\}$ is a descendant of $v_{i+1}$, to make sure that the gates considered later on along the path are different from $\{u_1,\ldots,u_k\}$.

What about product gates?
If $v_i$ is a product gate, with children $u_1,u_2$, we prove,
\begin{equation*}
\rank \left( M_{v_i}^{a_i,b_i} \right) \;\;\leq \;\;
\sum_{j=1}^{a_i} \rank \left( M_{u_2}^{a_i-j,b_i} \right) +
\sum_{j=1}^{b_i-1} \rank \left( M_{u_1}^{a_i,b_i-j} \right)
\end{equation*}
(In this proof, we use the property that every node $v$ in the modified circuit computes a polynomial in which all monomials are of degree at least~1).
Recall that we will have $v_{i+1} \in \{u_1,u_2\}$ and 
either $a_{i+1} = a_i$ and $b_{i+1} = b_i -j$, or $a_{i+1} = a_i -j$ and $b_{i+1} = b_i$, for some $j\geq 1$. Potentially, there are many terms here too, but note that on the right-hand side of the equation there are only two terms with $j=1$, and other terms result in a too rapid decrease in the degree $a_i+b_i$, which is costly for the circuit because it decreases the number of steps until the path reaches a leaf, and hence increases the average factor in which~$r_i$ needs to be reduced in each step along the path. Therefore, we will be able to choose 
$v_{i+1}, a_{i+1}, b_{i+1}$ such that either the ratio between $r_i$ and $r_{i+1}$ is less than a sufficiently large constant, or the degree decreases by more than~1, which means that the circuit needs more gates in later steps. We will choose 
$v_{i+1}, a_{i+1}, b_{i+1}$ to maximize a tradeoff between the ratio and the degree.

\subsubsection{Lower Bounds for Arithmetic Circuits over Rings}

Let $\F$ be a field.
Let $X=\{x_1,\ldots,x_n\}$, $Z=\{z_1,\ldots,z_n\}$ be sets of variables. 
Let $R = \F\langle Z \rangle$.
Let $f(x_1,\ldots,x_n) \in \F\langle X \rangle$.
We will consider non-commutative arithmetic circuits for $f$ over the field $\F$, on one hand, and on the other hand, arithmetic circuits over the ring $R$ for the induced polynomial function $f(x_1,\ldots,x_n)$, as a function from $R^n$ to $R$.

We prove that
if there exists an arithmetic circuit with set of input variables $X$,
over the ring $R$, that computes 
$f(x_1,\ldots,x_n)$ as a function from $R^n$ to $R$,
then there exists a non-commutative arithmetic circuit for $f$ over the field~$\F$ with the same number of sum gates, the same number of product gates, and the same size and depth.

In  an arithmetic circuit with a set of input variables $X$,
over the ring $\F\langle Z \rangle$, every leaf  is labeled with either a ring element $p \in \F\langle Z \rangle$, or an input variable $x_i$. 
This naturally
defines for each
node in the circuit
a non-commutative polynomial, in the ring of non-commutative polynomials $\F\langle Z, X \rangle$, that is computed by that node.

Let $g(z_1,\ldots, z_n, x_1,\ldots,x_n) \in \F\langle Z, X \rangle$ be the non-commutative polynomial computed by the output node of the circuit. Let $g^0(x_1,\ldots,x_n) \in \F\langle X \rangle$ be 
the restriction of $g$ to monomials in which none of the variables in $Z$ appear.
We can easily modify the circuit to be a non-commutative arithmetic circuit over $\F$ for the polynomial $g^0(x_1,\ldots,x_n)$:  
For every leaf  labeled with $p\in \F\langle Z \rangle$, we just replace~$p$ with its constant term $p^0 \in \F$.
It remains to prove that 
$g^0(x_1,\ldots,x_n) = f(x_1,\ldots,x_n)$. 

By the requirement from the circuit, we know that for every 
$h_1,\ldots, h_n \in \F\langle Z \rangle $, we have
$$g(z_1,\ldots, z_n, h_1,\ldots,h_n)= f(h_1,\ldots,h_n)$$
(as polynomials in $\F\langle Z \rangle$).
We show that this implies 
$$g^0(x_1,\ldots,x_n) = f(x_1,\ldots,x_n)$$ (as polynomials in $\F\langle X \rangle$).

The proof is by choosing $h_i = z_i^D$, where $D \in \N$ is larger than the degree of $g$, and noticing that, after this substitution, monomials in $g-g^0$ must cancel because their degree is not divisible by $D$, whereas monomials in $g^0$ and monomials in $f$ are of degree divisible by~$D$ (after the substitution).
Note that this argument does not prove that  $g-g^0$ is identically 0, just that it becomes 0 when substituting for every $i$, $x_i = z_i^D$.

We hence have,
$g^0(z_1^D,\ldots,z_n^D) = f(z_1^D,\ldots,z_n^D)$, 
and since $x_i \leftrightarrow z_i^D$ gives a bijection between words in 
$\{x_1,\ldots,x_n\}$ and words in $\{z_1^D,\ldots,z_n^D\}$,
we conclude that
$g^0(x_1,\ldots,x_n) = f(x_1,\ldots,x_n)$.

\section{Preliminaries}

\subsection{Non-Commutative Polynomials}

For a set of variables $X=\{x_1,\ldots,x_n\}$ and a field $\F$,
we denote by $\F\langle X \rangle = \F\langle x_1,\ldots,x_n \rangle $ the  ring of non-commutative  polynomials in (non-commuting) variables $x_1,\ldots,x_n$, with coefficients in $\F$.
A polynomial in $\F\langle X \rangle$ is a formal linear combination of words over the alphabet $\{x_1,\ldots,x_n\}$,
 with coefficients in 
$\F$.
The sum of two polynomials in $\F\langle X \rangle$ is their sum as linear combinations over $\F$, and their product is defined by defining the product of two monomials (words) as their concatenation and extending bilinearly to  all pairs of polynomials in $\F\langle X \rangle$.

\subsection{Arithmetic Circuits} \label{subsection: arithmetic circuits}

An arithmetic circuit with input variables
$x_1,\ldots,x_n$, over a ring $\R$,
is a directed acyclic graph as follows:
A node of in-degree~0 is called a leaf.
Every leaf in the circuit
is labeled with either an input variable
or a ring element.
A leaf
labeled with an input variable
is called an input gate.
Every non-leaf node is labeled with either $+$ or~$\times$,
in the first case the node is a sum  gate and in the second case a
product gate.
We assume that the in-degree of every product gate in the circuit is~2. 
The children of each product gate are ordered, and we refer to them 
as the left child and the  right child.
A sum gate may have an arbitrary in-degree greater than~0.
One node of out-degree~0 is called the output gate. 
We assume that only the output gate has out-degree 0, as other nodes with out-degree 0 can be removed. 
The circuit is called a formula if the underlying
graph is a (directed) tree.

The size of an arithmetic circuit is defined as the number of wires (edges) in it.
The depth of an arithmetic circuit is defined as the length of the longest directed path from a leaf to the output gate of the circuit.
In this paper, we will mainly be interested in the number of (non-scalar) product gates in the circuit, rather than the size of the circuit.

Note that we use this same definition of an arithmetic circuit for three different settings: general (commutative) arithmetic circuits over fields; non-commutative arithmetic circuits over fields; and, arithmetic circuits over (non-commutative) rings. The difference will be in the type of computation performed in each of the three cases, as we discuss next.

\subsection{Computation of Arithmetic Circuits over  Fields}

Given an arithmetic circuit with input variables $x_1,\ldots,x_n$, over a field $\F$,
each node in the circuit (and in particular the output node)
computes a polynomial in the ring of polynomials $\mathbb{F}[x_1,\ldots,x_n]$ as follows: A leaf just computes the input variable, or field element,
that labels it. 
A sum gate computes the sum of the polynomials computed by its children.
A product gate computes the product of the two polynomials computed by its children.
The polynomial computed by the circuit is the polynomial computed by the output gate.
 
\subsection{Computation of Non-Commutative Circuits over  Fields}

Non-commutative arithmetic circuits over fields are defined in the same way as standard arithmetic circuits over fields, except that they compute non-commutative polynomials rather than ordinary commutative polynomials. 
In other words, when we refer to an arithmetic circuit as non-commutative, we mean that the input variables are treated as non-commuting.

Given an arithmetic circuit with input variables $x_1,\ldots,x_n$, over a field $\F$, if we refer to the circuit as a non-commutative circuit, 
each node in the circuit (and in particular the output node)
computes a non-commutative polynomial in the ring of non-commutative polynomials $\mathbb{F}\langle x_1,\ldots,x_n \rangle$ as follows: A leaf just computes the input variable, or field element, 
that labels it. 
A sum gate computes the sum of the non-commutative polynomials computed by its children.
A product gate computes the product of the non-commutative polynomial computed by its left child and the non-commutative polynomial computed by its right child, in that order.
The non-commutative polynomial computed by the circuit is the non-commutative polynomial computed by the output gate.

\subsection{Computation of  Circuits over Non-Commutative Rings}

For an arithmetic circuit over a non-commutative ring $\R$, it is not reasonable to treat the input variables as commuting, since elements of $\R$ do not always commute.
On the other hand, treating the variables as fully non-commuting may fail to capture nontrivial relations satisfied in $\R$.
We therefore view  an arithmetic circuit over a non-commutative ring $\R$ as computing a {\em function} from $\R^n$ to $\R$, rather than a {\em polynomial}. 

Given an arithmetic circuit with input variables $x_1,\ldots,x_n$, over a ring $\R$, 
each node in the circuit (and in particular the output node)
computes a function from $\R^n$ to $\R$, as follows: A leaf labeled with an input variable $x_i$ computes the function that returns the $i$-th coordinate of the input. 
A leaf labeled with a ring element computes the function that returns the ring element 
that labels it. 
A sum gate computes the sum of the functions computed by its children.
A product gate computes the product of the function computed by its left child and the function computed by its right child, in that order.
The function computed by the circuit is the function computed by the output gate.

\section{Lower Bounds for Non-Commutative Circuits}

\subsection{Notation} \label{subsec:notation}

We denote by $\N$ the set of natural numbers, including 0.
Let $X=\{x_1,\ldots,x_n\}$ be a set of variables. Let $\F$ be a field. Assume that $n \geq2$. 

For $f \in \mathbb{F}\langle X \rangle$ and $r \in \N$, let $f^r \in \mathbb{F}\langle X \rangle$ be the homogeneous part of~$f$ of degree $r$. We have, $f = \sum_r f^r$. 
Let $f^{>0} = \sum_{r>0} f^r = f - f^0$.

For $f \in \mathbb{F}\langle X \rangle$ and  $a,b \in \N$, let
$M_f^{a,b}$ be the $n^a \times n^b$ matrix whose entry at
row $(i_1,\ldots,i_a)$ and column $(j_1,\ldots,j_b)$ is the coefficient in $\F$ of the monomial (word)
$x_{i_1}\ldots x_{i_a}x_{j_1}\ldots x_{j_b}$ in the polynomial $f$, where
$i_1,\ldots,i_a, j_1,\ldots,j_b \in \{1,\ldots,n\}$.

For a non-commutative circuit $C$ with a set of input variables  $X$, over the field $\F$, and
a node $v$ in $C$, let $f_v \in \mathbb{F}\langle X \rangle$ be the polynomial computed by the node~$v$. 
For simplicity, we denote the matrix $M_{f_v}^{a,b}$  by $M_{v}^{a,b}$.

A non-scalar product gate in a circuit is a product gate $v$, with children $v_1$ and $v_2$, such that both $f_{v_1}$ and $f_{v_2}$ are polynomials of degree at least~1. 
A scalar-product gate in a circuit is a product gate $v$, with children $v_1$ and $v_2$, such that at least one of  $f_{v_1}$ and $f_{v_2}$ are of degree~0. 

\subsection{Statement of Theorem~\ref{thm:1}} \label{subsec:thm1}

\begin{theorem} \label{thm:1}
Let $n \geq 2^{12}$ and $d\geq 2$. Assume without loss of generality that $d$ is even.
Let $C$ be a non-commutative arithmetic circuit with set of input variables  $X=\{x_1,\ldots,x_n\}$, over a field~$\F$.  Let $f \in \mathbb{F}\langle X \rangle$ be the non-commutative polynomial computed by $C$.
If the matrix~$M_f^{d/2,d/2}$ has full rank, then there are at least 
$\Omega(  d   \sqrt{n} )$ non-scalar product gates in $C$.
\end{theorem}

\begin{proof}

In the rest of this section we prove Theorem~\ref{thm:1}.

\subsection{Modifying the Circuit} \label{subsec:modifying}

Let $C$ be a non-commutative circuit with a set of input variables  $X$, over the field $\F$. Let $f \in \mathbb{F}\langle X \rangle$ be the polynomial computed by $C$. 
Assume that $f^{>0} \neq0$.
We will modify the circuit $C$ to a new circuit that computes $f^{>0}$ and has at most the same number of non-scalar product gates as $C$. The new circuit will have the following additional properties:
\begin{enumerate}
\item \label{p1}
There are no scalar-product gates in the circuit. Instead, 
every edge to a sum gate is labeled with a field element.
By convention, 
this element multiplies the output of the edge. 
We adopt this  convention to simplify the presentation.
By default, an edge is labeled with~1, unless said otherwise.
\item \label{p2}
Every node $v$ in the modified circuit computes a non-zero polynomial in which all monomials are of degree at least~1, that is, a polynomial without a constant term.
Formally, every  node $v$ has $f^0_v=0$, and $f_v=f^{>0}_v \neq 0$. 
\item \label{p3}
All edges from the leaves are to sum gates.
\item \label{p4}
The output gate is a sum gate.
\item \label{p5}
The gates are alternating.
That is, if $u$ is a product gate and $(u,v)$ is an edge then $v$ is a
sum gate, and if $u$ is a sum-gate and $(u,v)$ is an edge then $v$ is a product gate.
\end{enumerate}

Note that the most important property in this list is Property~\ref{p2}. The other properties are straightforward and are introduced in order to simplify the presentation. 
We will modify the circuit $C$ by the following steps.

\subsubsection*{Step 1: (Splitting each node into a degree-zero part and a positive-degree part):}

We modify the circuit to a new circuit such that for every node $v$ in the new circuit,
either $f^0_v=0$ or $f^{>0}_v=0$.
Note that $f^0_v=0$ implies $f_v=f^{>0}_v$ which means that all monomials in $f_v$ are of degree at least~1, while $f^{>0}_v=0$ implies
$f_v=f^0_v$ which
means that $f_v$ is just a constant field element and hence $v$ can be replaced by a leaf labeled with that field element.

The modification is done as follows. For every node $v$ in the original circuit, we split~$v$ into two nodes, $v^0$ and $v^{>0}$, in the new circuit, where $v^0$ computes the degree-zero part of the polynomial computed by $v$ and $v^{>0}$ computes the positive-degree  part of the polynomial computed by $v$, that is, $f_{v^0} = f^0_v$ and $f_{v^{>0}} = f^{>0}_v$.

We do that by induction over the circuit, starting from the leaves. The leaves are already in this form as each leaf $v$ computes either an input variable (in which case $f^0_v=0$) or a field element (in which case $f^{>0}_v=0$). 
Formally, if $v$ is a leaf labeled with an input variable we define $v^{>0}$ to be the same as  $v$, that is, $v^{>0}$ is a leaf labeled with the same input variable as $v$, and we add a leaf $v^0$ labeled with the field element~0, and if $v$ is a leaf labeled with a field element 
we define $v^0$ to be the same as  $v$, that is, $v^0$ is a leaf labeled with the same field element as $v$, and we add a leaf $v^{>0}$ labeled with the field element~0.

If $v$ is a sum gate, with children $v_1,\ldots,v_k$, we define $v^0$ to be a sum gate, with children $v^0_1,\ldots,v^0_k$, and $v^{>0}$ to be a sum gate, with children $v^{>0}_1,\ldots,v^{>0}_k$. We have,
$
f_{v^0}=\sum_{i=1}^k f_{v^0_i} 
$,
and 
$
f_{v^{>0}}=\sum_{i=1}^k f_{v^{>0}_i} 
$.

If $v$ is a product gate, with children $v_1,v_2$, we 
define $v^0$ to be a product gate, with children $v^0_1,v^0_2$,
so that we have,
$
f_{v^0}=f_{v^0_1} f_{v^0_2}
$.

As for $v^{>0}$, we separate into two cases, the case where $v$ is a scalar-product gate and the case where $v$ is a non-scalar product gate. If $v$ is a scalar-product gate, one of the polynomials $f_{v^{>0}_1}$ or $f_{v^{>0}_2}$ is equal to~0. Assume that $f_{v^{>0}_1}=0$ (the case where $f_{v^{>0}_2}=0$ is similar). We define in this case $v^{>0}$
to be a product gate, with children $v^0_1,v^{>0}_2$, so that we have,
$
f_{v^{>0}}=f_{v^0_1} f_{v^{>0}_2}
$. Note that in this case  $v^{>0}$ is a scalar-product gate, since $f_{v^0_1}$ is of degree~0.

If $v$ is a non-scalar product gate,
we need to have  
$f_{v^{>0}}=f_{v^0_1} f_{v^{>0}_2} + f_{v^0_2} f_{v^{>0}_1} + f_{v^{>0}_1} f_{v^{>0}_2}$,
so we just add a small circuit with output gate $v^{>0}$ that performs this computation. Note that in that small circuit there is one sum gate and 3 product gates. Two of these product gates are scalar-product gates and only one is a non-scalar product gate.

Altogether, in the modified circuit, the number of non-scalar product gates is the same as in the original circuit, as we had one non-scalar product gate in the modified circuit for each non-scalar product gate in the original circuit. Moreover, the total number of gates in the modified circuit is bounded by a small constant times the total number of gates in the original circuit, and the size and depth of the modified circuit are bounded by small constants times the size and depth of the original circuit, respectively. Recall that we only care about the number of non-scalar product gates. 

Since the next steps will not add non-scalar product gates to the circuit,
the final modified circuit will have at most the same number of non-scalar product gates as the original circuit.

\subsubsection*{Step 2: (Setting the output gate):}

If $v$ is the output gate of the original circuit, we take $v^{>0}$ to be the output gate of the modified circuit, and thus the modified circuit computes $f_{v^{>0}}=f^{>0}_v = f^{>0}$, as required. 
Since the next steps will not change the polynomial computed by the circuit, 
the final modified circuit will compute $f^{>0}$.

If~$v^{>0}$ is not a sum gate, 
we just add a new sum gate $o$ with in-degree 1 such that $v^{>0}$ is a child of~$o$, and we take $o$ to be the output gate.
Thus, we obtain Property~\ref{p4}.

\subsubsection*{Step 3: (Removing degree-zero nodes and disconnected parts):}

We simplify the circuit by the  following steps:

{\bf Step 3.1: (Replacing nodes that compute constant field elements, by leaves):} For every node $u$ such that $f^{>0}_u=0$, we know that $f_u$ is just a constant field element $c_u$. Hence, in the modified circuit, we replace each such node $u$ by a leaf labeled with the field element $c_u$ (and remove all edges to the node $u$).
Recall that by Step~1 every other node $v$ in the modified circuit has 
$f_v = f^{>0}_v \neq 0$.

{\bf Step 3.2: (Removing disconnected parts):} We remove from the circuit every node that is not connected to the output gate by a directed path (and all edges to and from that node). In particular, this removes all nodes of out-degree~0 that are not the output node.

{\bf Step 3.3: (Removing scalar-product gates):} Every scalar-product gate $v$ in the modified circuit multiplies a node $u$ and a leaf $w$ labeled with a field element~$c$. For each such scalar-product gate, 
we remove the gate~$v$ and the edge  $(w,v)$ and the leaf $w$ if its new out-degree is~0, and we directly connect~$u$ to all parents of $v$.
We place the field element~$c$ as a label on 
every edge from $u$ to a parent of~$v$ (multiplying any field element that already labels that edge).
By convention, the element that labels an edge multiplies the output of the edge.
We adopt this  convention to simplify the presentation,
as now there are no scalar-product gates in the circuit (except for these labels).
By default, an edge is labeled with~1, unless said otherwise.

Moreover, for any edge $e=(a,b)$, labeled with a field element $c$, such that~$b$ is a product gate, we can further move the label $c$ of the edge $e$ to every edge $e'$ from the gate~$b$ (multiplying any field element that already labels $e'$). We do that for every such edge in the circuit, going over the circuit bottom up, so that in the end, only edges to sum gates are labeled with field elements. Thus, we obtain Property~\ref{p1}.

{\bf Step 3.4: (Removing all leaves labeled with field elements):}
We can now remove from the circuit every leaf $u$ that is labeled with a field element, together with all the edges from $u$. Note that $u$ is not connected to a product gate, because we removed all scalar-product gates in Step~3.3. 
For any sum gate $v$,  since $v$ is not a leaf, it has
$f_v = f^{>0}_v$  (since other nodes were replaced by leaves in Step 3.1).
Hence $f_v$ doesn't change if we remove all leaves labeled with field elements that are children of $v$. Thus, by a bottom up induction over the circuit, the polynomial computed by every gate of the circuit doesn't change when we remove from the circuits all leaves that are labeled with field elements.

Since in Step~3.1 we replaced each node $u$ such that $f^{>0}_u=0$ by a leaf labeled with a field element, and in Step~3.4 we removed all remaining leaves labeled with field elements, we have removed from the circuit all nodes $u$ such that $f^{>0}_u=0$.
Recall that by Step~1, every other node $v$ in the modified circuit has 
$f_v = f^{>0}_v \neq 0$. Thus, we obtain Property~\ref{p2}.

\subsubsection*{Step 4: (Alternating the gates):} 

For every edge from a leaf $w$ to a product gate $v$, we add a sum gate $u$ with in-degree 1 in between $w$ and $v$, that is, we create a sum gate $u$ such that $w$ is a child of $u$ and $u$ is a child of $v$. Thus, we obtain Property~\ref{p3}.

Finally, we ensure that the gates are alternating. For any edge $(u,v)$ such that
$u,v$ are both product-gates, we add a sum gate $w$ with in-degree 1 in between $u$ and $v$, that is, we create a sum gate $w$ such that $w$ is a child of $v$ and $u$ is a child of $w$.
For any edge $(u,v)$ such
that $u,v$ are both sum-gates, we connect each  child $w$ of $u$ directly to $v$ and remove the edge $(u,v)$. The label placed on the edge $(w,v)$ is the product of the label on $(w,u)$ and the one on $(u,v)$. Thus, we obtain Property~\ref{p5}.

We repeat Step~3.2 (removing disconnected parts) if needed.

\paragraph{} 

For the rest of the proof we assume that the circuit $C$ is already the modified circuit.

\subsection{Bounding the Rank at a Sum Gate}

If $v$ is a sum gate, with children $v_1,\ldots,v_k$, and each edge $(v_i,v)$ is labeled with the field element $c_i$, we have,
$$
f_v=\sum_{i=1}^k c_i f_{v_i} 
$$
Hence, for any $a,b \in \N$,
$$
M_{v}^{a,b}=\sum_{i=1}^k c_i M_{v_i}^{a,b} 
$$
Hence, by subadditivity of the rank,
\begin{equation}\label{eq:rank-sum}
\rank \left( M_{v}^{a,b} \right) \;\;\leq\;\; \sum_{i=1}^k \rank \left( M_{v_i}^{a,b} \right)
\end{equation}

\subsection{Bounding the Rank at a Product Gate}

If $v$ is a product gate, with children $v_1,v_2$, we have 
$$
f_v=f_{v_1} f_{v_2}
$$
Hence, for any $r \in \N$,
$$
f^{r}_v=\sum_{i=0}^{r} f^i_{v_1} f^{r-i}_{v_2}
$$
Note that $v$ is a non-scalar product gate (as we removed all scalar-product gates and treat them as labels on the edges - see Property~\ref{p2} in Subsection~\ref{subsec:modifying}).
Therefore, by Property~\ref{p1} in Subsection~\ref{subsec:modifying}, we have
$f^0_{v_1}=0$ and $f^0_{v_2}=0$. 
Hence, for any $r \in \N$,
$$
f^{r}_v=\sum_{i=1}^{r-1} f^i_{v_1} f^{r-i}_{v_2}
$$
Hence, for any $a,b \in \N$,
$$
f^{a+b}_v \;\;=\;\; \sum_{i=1}^{a+b-1} f^i_{v_1} f^{a+b-i}_{v_2} \;\;=\;\;
\sum_{i=1}^{a} f^i_{v_1} f^{a+b-i}_{v_2} + \sum_{i=a+1}^{a+b-1} f^{i}_{v_1} f^{a+b-i}_{v_2} 
$$
$$ \;\;=\;\;
\sum_{i=1}^{a} f^i_{v_1} f^{a+b-i}_{v_2} + \sum_{i=1}^{b-1} f^{a+i}_{v_1} f^{b-i}_{v_2} \;\;=\;\;
\sum_{i=1}^{a} f^i_{v_1} f^{a+b-i}_{v_2} + \sum_{i=1}^{b-1} f^{a+b-i}_{v_1} f^{i}_{v_2}
$$
Hence, for any $a,b \in \N$,
$$
M_{v}^{a,b} \;\;=\;\; \sum_{i=1}^a \left( M_{v_1}^{i,0} \otimes M_{v_2}^{a-i,b} \right) +
\sum_{i=1}^{b-1} \left( M_{v_1}^{a,b-i} \otimes M_{v_2}^{0,i} \right)
$$
(where $\otimes$ denotes Kronecker product (tensor product)).

Note that for every $i$, we have $\rank(M_{v_1}^{i,0}) \leq 1$ as $M_{v_1}^{i,0}$ has one column, and
$\rank(M_{v_2}^{0,i}) \leq 1$ as $M_{v_2}^{0,i}$ has one row.
Recall also that for any two matrices $A,B$, we have $\rank(A\otimes B) = \rank(A)  \rank(B)$.
Hence, by subadditivity of the rank, we have,
\begin{equation}\label{eq:rank-product}
\rank \left( M_{v}^{a,b} \right) \;\;\leq \;\;
\sum_{i=1}^a \rank \left( M_{v_2}^{a-i,b} \right) +
\sum_{i=1}^{b-1} \rank \left( M_{v_1}^{a,b-i} \right)
\end{equation}

\subsection{Defining a Backward Path from the Output}

For simplicity and without loss of generality, we assume that $d$ is even. Let $v_0$ be the output gate of the circuit. We will define a path $(v_0,v_1,\ldots,v_t)$ that goes backward on the circuit, starting from the output gate. 
For every~$i$, we will define $v_{i+1}$ to be one of the children of~$v_i$, according to the rules specified below, until the path stops.

For every $v_i$, we will also define $a_i,b_i \in \N$. 
Let $a_0,b_0 = d/2$.
For every $i$, we will have,
$a_{i+1} \leq a_i$ and $b_{i+1} \leq b_i$. 
For every $i$, we define 
$$
r_i \; = \; \rank \big( M_{v_i}^{a_i,b_i} \big)
$$
If $r_i \leq 1$, the path stops.
Note that if $a_i=0$ or $b_i=0$ then $r_i \leq 1$ and the path stops.
Also, if $v_i$ is a leaf then $r_i \leq 1$ and the path stops.

If $v_i$ is a product gate, we will also define a set $S_i$ of product gates in the circuit. Every gate in $S_i$ will be a child of $v_{i-1}$ and will not be a descendant of $v_i$. 
Note that $v_0$ is a sum gate so we do not need to define~$S_0$.

We will use two constants: 
$$c=64$$
and 
$$\alpha=\tfrac{1}{4}$$

\subsubsection*{Sum Gates:}

Assume that  $v_i$ is a sum gate, with a set of children 
$N_i = \{u_1,\ldots,u_m\}$. 
We set $$a_{i+1} = a_i \;\; , \;\; b_{i+1} = b_i$$
Let
$$
t_i \;=\; \tfrac{c}{\sqrt{n}} \cdot r_i \;=\; \tfrac{c}{\sqrt{n}} \cdot \rank \big( M_{v_i}^{a_i,b_i} \big)
$$

\paragraph{\bf Rule-1:}
If there exists  $u \in N_i$ with 
$\rank \big( M_{u}^{a_i,b_i} \big) \geq t_i $, we set 
$v_{i+1} = u$ 
and we set $S_{i+1} = \emptyset$.
We say in this case that $v_{i+1}$ was chosen by Rule-1, to keep a record of how $v_{i+1}$ was chosen.

Note that we have,
\begin{equation} \label{eq:I1}
r_{i+1} \;=\;\rank \big( M_{u}^{a_i,b_i} \big) \;\geq\; t_i \;=\; \tfrac{c}{\sqrt{n}} \cdot r_i
\end{equation}

\paragraph{\bf Rule-2:}
Otherwise,
for every $k \in \N$,
let  
$$ N_{i,k} \;=\; \left\{ u \in N_i \;\; \middle| \;\; \alpha^{k+1}  \cdot t_i \; \leq \; \rank \big( M_{u}^{a_i,b_i} \big) \;<\; \alpha^k \cdot t_i \right\}$$
\begin{equation} \label{eq:rik}
r_{i,k} \;=\; \sum_{u \in N_{i,k}}  \rank \big( M_{u}^{a_i,b_i} \big) 
\;<\; |N_{i,k}| \cdot \alpha^k \cdot t_i
\end{equation}
By Equation~\eqref{eq:rank-sum},
$$   \sum_{k \in \N} r_{i,k}  \;\geq\; r_i
$$
Hence, there exists $k \in \N$, such that,
$$   r_{i,k}  \;\geq\; 2^{-(k+1)} \cdot r_i
$$
Hence, for that $k$, by Equation~\eqref{eq:rik},
$$
|N_{i,k}| \cdot \alpha^k \cdot t_i \; > \; 2^{-(k+1)} \cdot r_i
$$
Hence, by the definition of $t_i$,
$$
|N_{i,k}| \cdot \alpha^k \cdot \tfrac{c}{\sqrt{n}} \; > \; 2^{-(k+1)} 
$$
That is,
\begin{equation} \label{nik}
|N_{i,k}|    \;> \; \tfrac{\sqrt{n}}{2c} \cdot (2\alpha)^{-k} 
\end{equation}

Since the circuit is an acyclic graph, there must be a node $u \in N_{i,k}$ such that none of the nodes $u' \in N_{i,k}$ is a descendant of $u$. We set 
$v_{i+1} = u$,  and we set $S_{i+1} = N_{i,k}$.
We say in this case that $v_{i+1}$ was chosen by Rule-2, to keep a record of how $v_{i+1}$ was chosen, and
we define $k_{i} :=k$, to record the $k$ that was used to choose $v_{i+1}$.

Note that all the gates in $S_{i+1}$ are product gates. This is true because all of them are children of a sum gate, so they can be either  product gates or leaves. However, they cannot be leaves because for every $u' \in S_{i+1}$, we have 
$$\rank \big( M_{u'}^{a_i,b_i} \big) \;\;\geq\;\; \alpha^{k+1}  \cdot t_i \;\;>\;\; 0 $$
where $a_i,b_i >0$, while for a leaf $w$ we always have $M_{w}^{a_i,b_i} =0$.

Note that we have,
\begin{equation} \label{eq:I2}
r_{i+1} \;=\;\rank \big( M_{u}^{a_i,b_i} \big) \;\geq\; \alpha^{k+1}  \cdot t_i \;=\; \alpha^{k+1}  \cdot \tfrac{c}{\sqrt{n}} \cdot r_i
\end{equation}

\subsubsection*{Product Gates:}

Assume that  $v_i$ is a product gate, with children $u_1,u_2$.
By Equation~\eqref{eq:rank-product},
$$
\sum_{j=1}^{a_i} \rank \left( M_{u_2}^{a_i-j,b_i} \right) +
\sum_{j=1}^{b_i-1} \rank \left( M_{u_1}^{a_i,b_i-j} \right)
\;\;\geq \;\;
\rank \left( M_{v_i}^{a_i,b_i} \right) 
\;\;=\;\; r_i
$$
Therefore, there must exist either $j \in \{1,\ldots,a_i\}$, such that,
$$
\rank \left( M_{u_2}^{a_i-j,b_i} \right) \;\;\geq \;\; 2^{-(j+1)} \cdot r_i
$$
or $j \in \{1,\ldots,b_i-1\}$, such that,
$$
\rank \left( M_{u_1}^{a_i,b_i-j} \right) \;\;\geq \;\; 2^{-(j+1)} \cdot r_i
$$
In the first case, we set 
$v_{i+1} = u_2$,  and we set $a_{i+1} = a_i-j$, $b_{i+1} = b_i$.
and
we define $j_{i} :=j$, to record the $j$ that was used to choose $a_{i+1}$.
In the second case, we set 
$v_{i+1} = u_1$,  and we set $a_{i+1} = a_i$, $b_{i+1} = b_i-j$.
and
we define $j_{i} :=j$, to record the $j$ that was used to choose~$b_{i+1}$.

Note that in both cases we have,
\begin{equation} \label{eq:I3}
r_{i+1}   \;\;\geq \;\; 2^{-(j+1)} \cdot r_i  \;\;\geq \;\; 2^{-2j} \cdot r_i
\end{equation}
and $a_{i+1} + b_{i+1} = a_i+b_i-j$.

\subsection{Analysis}

Let $(v_0,v_1,\ldots,v_t)$ be the backward path that we defined.
Assume that $r_0= n^{d/2}$ and we stopped when $r_t \leq 1$. Thus,
\begin{equation} \label{eq:a0}
\frac{r_t}{r_0} \;=\; \prod_{i=0}^{t-1} \left(\frac{r_{i+1}}{r_{i}}\right) \;\leq\; n^{-d/2}
\end{equation}

Let $I_1 \subseteq \{0,\ldots,t-1\}$ be the set of indices $i$ where $v_i$ is a sum gate and $v_{i+1}$ was chosen by Rule-1.
Let $I_2 \subseteq \{0,\ldots,t-1\}$ be the set of indices $i$ where $v_i$ is a sum gate and $v_{i+1}$ was chosen by Rule-2.
Let $I_3 \subseteq \{0,\ldots,t-1\}$ be the set of indices $i$ where $v_i$ is a product gate.
We have,
\begin{equation} \label{eq:a1}
\frac{r_t}{r_0} \;=\; \prod_{i\in I_1} \left(\frac{r_{i+1}}{r_{i}}\right) \cdot
\prod_{i\in I_2} \left(\frac{r_{i+1}}{r_{i}}\right) \cdot
\prod_{i\in I_3} \left(\frac{r_{i+1}}{r_{i}}\right)
\end{equation}

For $I_1$, we have by Equation~\eqref{eq:I1},
\begin{equation} \label{eq:a2}
\prod_{i\in I_1} \left(\frac{r_{i+1}}{r_{i}}\right) \;\geq\;
\prod_{i\in I_1} \tfrac{c}{\sqrt{n}}
\;=\;
\left(\tfrac{c}{\sqrt{n}} \right)^{|I_1|}
\;\geq\;
\left(\tfrac{\alpha c}{\sqrt{n}} \right)^{|I_1|}
\end{equation}

For $I_2$, we have by Equation~\eqref{eq:I2},
\begin{equation} \label{eq:a3}
\prod_{i\in I_2} \left(\frac{r_{i+1}}{r_{i}}\right) \;\geq\;
\prod_{i\in I_2} \left( \alpha^{k_i+1}  \cdot \tfrac{c}{\sqrt{n}} \right)
\;=\;
\left(\tfrac{\alpha c}{\sqrt{n}} \right)^{|I_2|} \cdot \alpha^{\sum_{i\in I_2}k_i} 
\end{equation}

For $I_3$, we have by Equation~\eqref{eq:I3},
\begin{equation} \label{eq:a4}
\prod_{i\in I_3} \left(\frac{r_{i+1}}{r_{i}}\right) \;\geq\;
\prod_{i\in I_3} 2^{-2j_i} \;=\; 2^{-2\sum_{i\in I_3} j_i} \;\geq\; 4^{-d}
\end{equation}
where the last inequality holds since $\sum_{i\in I_3} j_i \leq d$, which is true since $a_0+b_0=d$ and $a_t+b_t\geq 0$, and since when $v_i$ is a sum gate, we have 
$a_{i+1} + b_{i+1} = a_i+b_i$, and when $v_i$ is a product gate, we have
$a_{i+1} + b_{i+1} = a_i+b_i-j_i$.

Note that since in the path $(v_0,v_1,\ldots,v_t)$ the gates are alternating, and since 
$a_0+b_0=d$ and $a_t+b_t\geq 1$, and since when~$v_i$ is a sum gate, we have 
$a_{i+1} + b_{i+1} = a_i+b_i$, and when~$v_i$ is a product gate, we have
$a_{i+1} + b_{i+1} \leq a_i+b_i-1$,
we have at most $d$ sum gates in $(v_0,v_1,\ldots,v_t)$, and hence, $|I_1|+|I_2| \leq d$.
($a_t+b_t\geq 1$ because if $a_i=0$ or $b_i=0$ the path stops, and at each step only one of them can decrease).

Thus, by Equations~\eqref{eq:a0}~\eqref{eq:a1}~\eqref{eq:a2}~\eqref{eq:a3}~\eqref{eq:a4},
$$
n^{-d/2}
\;\geq\;
\frac{r_t}{r_0} 
\;\geq\;
\left(\tfrac{\alpha c}{\sqrt{n}} \right)^{|I_1|+|I_2|} \cdot 
\alpha^{\sum_{i\in I_2}k_i} \cdot 4^{-d}
\;\geq\;
\left(\tfrac{\alpha c/4}{\sqrt{n}} \right)^{d} \cdot 
\alpha^{\sum_{i\in I_2}k_i} 
$$
That is,
$$
\left(\alpha c/4 \right)^{d} \cdot 
\alpha^{\sum_{i\in I_2}k_i}
\;\leq\;
1
$$
Substituting, $c=64$ and $\alpha = 1/4$, we conclude,
$$
4^{d} \cdot 
\left(\tfrac{1}{4}\right)^{\sum_{i\in I_2}k_i}
\;\leq\;
1
$$
and hence,
$$
\sum_{i\in I_2}k_i \geq d
$$
(which also implies that $I_2$ is not empty).

Let 
$$
S = \bigcup_{i\in I_2} S_{i+1}
$$
Recall that for every $i\in I_2$, all the nodes in $S_{i+1}$ are (non-scalar) product gates. Recall that by the way that $v_{i+1}$ was chosen, by Rule-2,  none of the nodes in $S_{i+1}$ are descendants of~$v_{i+1}$. (We don't count $v_{i+1}$ as a descendant of itself). On the other hand, for every 
$i'>i$, such that $i'\in I_2$, all the nodes in $S_{i'+1}$ are descendants of $v_{i+1}$. Hence, the sets $\{S_{i+1}\}_{i\in I_2}$ are disjoint.
Thus,
$$
|S| = \sum_{i\in I_2} |S_{i+1}|
$$

By Equation~\eqref{nik}, (substituting $\alpha = 1/4$), for every $i\in I_2$,
$$
|S_{i+1}|    \;> \; \tfrac{\sqrt{n}}{2c} \cdot 2^{k_i}
$$
Thus, by convexity of the function $2^x$, and since $\sum_{i\in I_2}k_i \geq d$, 
and since for every $x$, $2^x > x$,
$$
|S|  \;> \; \tfrac{\sqrt{n}}{2c} \cdot \sum_{i\in I_2}  2^{k_i}
\;\geq \; \tfrac{\sqrt{n}}{2c} \cdot |I_2| \cdot 2^{\sum_{i\in I_2}  {k_i/|I_2|}}
\;\geq \; \tfrac{\sqrt{n}}{2c} \cdot |I_2| \cdot 2^{  {d/|I_2|}}
\;> \;
\tfrac{\sqrt{n}}{2c} \cdot |I_2| \cdot {d/|I_2|} \;=\;
\tfrac{\sqrt{n}}{2c} \cdot d
$$

Thus, $S$ is a set of at least $\Omega(  d   \sqrt{n} )$  non-scalar product gates in $C$.
This completes the proof of Theorem~\ref{thm:1}.
\end{proof}

\section{Lower Bounds for Arithmetic Circuits over Rings}

We observe a connection between lower bounds for non-commutative arithmetic circuits over fields
and lower bounds for arithmetic circuits over non-commutative rings.
The connection follows by the following lemma.

Given a field $\F$ and sets of variables
$X=\{x_1,\ldots,x_n\}$, $Z=\{z_1,\ldots,z_n\}$, 
the lemma considers non-commutative polynomials in
$\F\langle X \rangle$, $\F\langle Z \rangle$, and $\F\langle Z, X \rangle = \F\langle Z\cup X \rangle$.

\begin{lemma} \label{lemma:poly}
Let $\F$ be a field.
Let $X=\{x_1,\ldots,x_n\}$, $Z=\{z_1,\ldots,z_n\}$ be sets of variables. 
Let $f(x_1,\ldots,x_n) \in \F\langle X \rangle$.
Let $g(z_1,\ldots, z_n, x_1,\ldots,x_n) \in \F\langle Z, X \rangle$.

Assume that for every 
$h_1,\ldots, h_n \in \F\langle Z \rangle $, we have,
$$g(z_1,\ldots, z_n, h_1,\ldots,h_n)= f(h_1,\ldots,h_n)$$
(as polynomials in $\F\langle Z \rangle$).

Let $g^0(x_1,\ldots,x_n) \in \F\langle X \rangle$ be 
the restriction of $g$ to monomials in which none of the variables in $Z$ appear, that is,
$g^0(x_1,\ldots,x_n) = g(0,\ldots, 0, x_1,\ldots,x_n)$.
Then 
$$g^0(x_1,\ldots,x_n) = f(x_1,\ldots,x_n)$$ (as polynomials in $\F\langle X \rangle$).
\end{lemma}

\begin{proof}
Let $g^0(x_1,\ldots,x_n) \in \F\langle X \rangle$ be 
the restriction of $g$ to monomials in which none of the variables in $Z$ appear. 
Let $g'(z_1,\ldots, z_n,x_1,\ldots,x_n) \in \F\langle Z,X \rangle$ be 
the restriction of $g$ to monomials in which at least one of the variables in $Z$ appear.
Thus,
$$g = g^0 + g'$$

Let $D \in \N$ be larger than the degree of $g$. For every 
$i \in \{1,\ldots, n\}$, 
let $h_i = z_i^D$. By the assumption,
$$g(z_1,\ldots, z_n, z_1^D,\ldots,z_n^D)= f(z_1^D,\ldots,z_n^D)$$
Hence, 
$$g'(z_1,\ldots, z_n, z_1^D,\ldots,z_n^D) + g^0(z_1^D,\ldots,z_n^D)
= f(z_1^D,\ldots,z_n^D)$$

Note that all monomials of $f(z_1^D,\ldots,z_n^D)$ and $g^0(z_1^D,\ldots,z_n^D)$  are of degree divisible by~$D$.
On the other hand,
in every monomial of $g'(z_1,\ldots, z_n,x_1,\ldots,x_n)$,
the number of times that variables in $Z$ appear is larger than 0 and smaller than $D$.
Therefore, when substituting $x_1 =z_1^D, \ldots, x_n = z_n^D$, each of the monomials of  $g'(z_1,\ldots, z_n,x_1,\ldots,x_n)$ gives a monomial of degree not divisible by $D$.
Thus, all monomials of $g'(z_1,\ldots, z_n, z_1^D,\ldots,z_n^D)$ are of degree not divisible by $D$,
and since in  $f(z_1^D,\ldots,z_n^D)$ and $g^0(z_1^D,\ldots,z_n^D)$ all monomials are of degree divisible by $D$, all monomials of $g'(z_1,\ldots, z_n, z_1^D,\ldots,z_n^D)$ must cancel, and we get $g'(z_1,\ldots, z_n, z_1^D,\ldots,z_n^D)=0$.

Thus,
$$g^0(z_1^D,\ldots,z_n^D) = f(z_1^D,\ldots,z_n^D)$$
and hence,
$$g^0(x_1,\ldots,x_n) = f(x_1,\ldots,x_n)$$
as $x_i \leftrightarrow z_i^D$ gives a bijection between words in 
$\{x_1,\ldots,x_n\}$ and words in $\{z_1^D,\ldots,z_n^D\}$.
\end{proof}

\subsection{The Connection}

Let $\F$ be a field.
Let $X=\{x_1,\ldots,x_n\}$, $Z=\{z_1,\ldots,z_n\}$ be sets of variables. 
Let $R = \F\langle Z \rangle$.
Let $f(x_1,\ldots,x_n) \in \F\langle X \rangle$.
We will consider non-commutative arithmetic circuits for $f$ over the field $\F$, on one hand, and on the other hand, arithmetic circuits over the ring $R$ for the induced polynomial function $f(x_1,\ldots,x_n)$, as a function from $R^n$ to $R$.

\begin{theorem}\label{thm:2}
Let $\F$ be a field.
Let $X=\{x_1,\ldots,x_n\}$, $Z=\{z_1,\ldots,z_n\}$ be sets of variables. 
Let $R = \F\langle Z \rangle$.
Let $f(x_1,\ldots,x_n) \in \F\langle X \rangle$.
Assume that there exists an arithmetic circuit with set of input variables $X$,
over the ring $R$, that computes 
$f(x_1,\ldots,x_n)$ as a function from $R^n$ to $R$.
Then, there exists a non-commutative arithmetic circuit for $f$ over the field~$\F$ with the same number of sum gates, the same number of product gates, and the same size and depth. 
\end{theorem}

\begin{proof}
Let $C$ be an arithmetic circuit with set of input variables $X$,
over the ring $\F\langle Z \rangle$, that computes 
$f(x_1,\ldots,x_n)$ as a function from $\F\langle Z \rangle^n$ to $\F\langle Z \rangle$.

Every leaf in $C$ is labeled with either a ring element $p \in \F\langle Z \rangle$, or an input variable $x_i$. 
We 
can define for each
node in $C$
a non-commutative polynomial in the ring of non-commutative polynomials $\F\langle Z, X \rangle$ that is computed by that node, as usual: 
A leaf just computes the input variable $x_i$, or ring element $p\in \F\langle Z \rangle$, 
that labels it. 
A sum gate computes the sum of the  polynomials computed by its children.
A product gate computes the product of the polynomial computed by its left child and the polynomial computed by its right child, in that order.
The polynomial computed by the circuit is the polynomial computed by the output gate.

Let $g(z_1,\ldots, z_n, x_1,\ldots,x_n) \in \F\langle Z, X \rangle$ be the non-commutative polynomial computed by~$C$. Since $C$ computes 
$f(x_1,\ldots,x_n)$ as a function from $\F\langle Z \rangle^n$ to $\F\langle Z \rangle$,
for every 
$h_1,\ldots, h_n \in \F\langle Z \rangle $, we have,
$$g(z_1,\ldots, z_n, h_1,\ldots,h_n)= f(h_1,\ldots,h_n)$$
(as elements in $\F\langle Z \rangle$).

Let $g^0(x_1,\ldots,x_n) \in \F\langle X \rangle$ be 
the restriction of $g$ to monomials in which none of the variables in $Z$ appear. By Lemma~\ref{lemma:poly},
$$g^0(x_1,\ldots,x_n) = f(x_1,\ldots,x_n)$$ (as polynomials in $\F\langle X \rangle$).

We can easily modify $C$ to be a non-commutative arithmetic circuit over $\F$ for the polynomial $g^0(x_1,\ldots,x_n)$ as follows. 
For every leaf in $C$, labeled with $p\in \F\langle Z \rangle$, we just replace $p$ with its constant term $p^0 \in \F$.
By induction over the circuit, for every node $v$ in $C$ that computes a non-commutative polynomial $g_v(z_1,\ldots, z_n, x_1,\ldots,x_n) \in \F\langle Z, X \rangle$, the modified circuit computes the restriction of $g_v$ to monomials in which none of the variables in $Z$ appear.
Hence, the output node computes $g^0$.
\end{proof}

\begin{corollary} \label{cor:1}
Let $n,d\geq 2$. Assume without loss of generality that $d$ is even.
Let $\F$ be a field.
Let $X=\{x_1,\ldots,x_n\}$, $Z=\{z_1,\ldots,z_n\}$ be sets of variables. 
Let $R = \F\langle Z \rangle$.
Let $f \in \mathbb{F}\langle X \rangle$ be such that
the matrix~$M_f^{d/2,d/2}$ has full rank (see Subsection~\ref{subsec:notation}).
Let $C$ be an arithmetic circuit with set of input variables $X$,
over the ring $R$, that computes 
$f(x_1,\ldots,x_n)$ as a function from $R^n$ to $R$.
Then, there are at least 
$\Omega(  d   \sqrt{n} )$  product gates in $C$.
\end{corollary}

\begin{proof}
Follows by Theorem~\ref{thm:1} and Theorem~\ref{thm:2}.
\end{proof}

\end{document}